\documentclass[preprint]{sigplanconf}
\usepackage[varg]{txfonts}

\usepackage{amsmath}
\usepackage{algorithm,algpseudocode}
\usepackage{amsfonts}
\usepackage{booktabs}
\usepackage{graphicx}
\usepackage{color}
\usepackage[colorlinks]{hyperref}
\usepackage{url}
\usepackage{paralist}
\usepackage{balance}
\usepackage{subfig}
\usepackage{ramkidefns}
      
\usepackage{tikz}
\usepackage{pgfplots}
\graphicspath{{./figs/}}

\newcommand{\nnz}{\text{nnz}}
\newcommand{\NaiveAlg}{Naive}
\newcommand{\ParNMF}{HPC-NMF}

\begin{document}

\newcommand{\lt}{\left}
\newcommand{\rt}{\right}

\setlength{\pdfpageheight}{\paperheight}
\setlength{\pdfpagewidth}{\paperwidth}

\titlebanner{}\preprintfooter{}
\title{A High-Performance Parallel Algorithm for Nonnegative Matrix Factorization}

\authorinfo{Ramakrishnan Kannan}{Georgia Tech}{rkannan@gatech.edu}
\authorinfo{Grey Ballard}{Sandia National Laboratories}{gmballa@sandia.gov}
\authorinfo{Haesun Park}{Georgia Tech}{hpark@cc.gatech.edu}

\maketitle

\begin{abstract}
Non-negative matrix factorization (NMF) is the problem of determining two non-negative low rank factors $\WW$ and $\HH$, for the given input matrix $\AA$, such that $\AA \approx \WW \HH$.  NMF is a useful tool for many applications in different domains such as topic modeling in text mining, background separation in video analysis, and community detection in social networks.  Despite its popularity in the data mining community, there is a lack of efficient parallel software to solve the problem for big datasets.  Existing distributed-memory algorithms are limited in terms of performance and applicability, as they are implemented using Hadoop and are designed only for sparse matrices.  

We propose a distributed-memory parallel algorithm that computes the factorization by iteratively solving alternating non-negative least squares (NLS) subproblems for $\WW$ and $\HH$.  To our knowledge, our algorithm is the first high-performance parallel algorithm for NMF.  It maintains the data and factor matrices in memory (distributed across processors), uses MPI for interprocessor communication, and, in the dense case, provably minimizes communication costs (under mild assumptions).  As opposed to previous implementations, our algorithm is also flexible: (1) it performs well for dense and sparse matrices, and (2) it allows the user to choose from among multiple algorithms for solving local NLS subproblems within the alternating iterations. We demonstrate the scalability of our algorithm and compare it with baseline implementations, showing significant performance improvements.
\end{abstract} \section{Introduction}

Non-negative Matrix Factorization (NMF) is the problem of finding two low rank factors $\WW\in \Rnplus{m\times k}$ and $\HH\in \Rnplus{k \times n}$ for a given input matrix  $\AA\in \Rn{m\times n}$, such that $\AA \approx \WW \HH$, where $k \ll min(m,n)$.
Formally, NMF problem \cite{seung2001algorithms} can be defined as \SplitN{\label{eqn:original NMF}}{
\min_{\WW \geq 0,\HH \geq 0} & \quad f(\WW,\HH) \equiv \|\AA-\WW\HH\|_F^2.\\
}

NMF is widely used in data mining and machine learning as a dimension reduction and factor analysis method. 
It is a natural fit for many real world problems as the non-negativity is inherent in many representations of real-world data and
the resulting low rank factors are expected to have natural interpretation. The applications of NMF range from text mining \cite{pauca2004text},  computer vision \cite{hoyer2004non}, bioinformatics \cite{kim2007sparse}, to blind source separation  \cite{cichocki2009nonnegative}, unsupervised clustering \cite{kuang2012symmetric}  and many other areas.
For typical problems today, $m$ and $n$ can be on the order of thousands to millions or more, and $k$ is typically less than 100.

There is a vast literature
on algorithms for NMF and their convergence properties \cite{kim2013nonnegative}.   
The commonly adopted NMF algorithms are -- (i) Multiplicative Update (MU) \cite{seung2001algorithms} (ii) Hierarchical Alternative Least Squares (HALS) \cite{cichocki2009nonnegative} (iii) Block Principal Pivoting (BPP) \cite{kim2011fast}, and (iv) Stochastic Gradient Descent (SGD) Updates \cite{gemulla2011large}. 
As described in Equation \ref{eqn:two block}, most of the algorithms in NMF literature are based on Alternating Non-negative Least Squares (ANLS) scheme that iteratively optimizes each of the low rank factors $\WW$ and $\HH$ while keeping the other fixed. 
It is important to note that in such iterative alternating minimization techniques, each subproblem is a constrained convex optimization problem. 
Each of these subproblems is then solved using standard optimization techniques such as projected gradient, interior point, etc., and a detailed survey for solving this constrained convex optimization problem can be found in \cite{xiong2013survey,kim2013nonnegative}. 
In this paper, for solving the subproblems, our implementation uses a fast active-set based method called Block Principal Pivoting (BPP) \cite{kim2011fast}, but the parallel algorithm proposed in this paper can be easily extended for other algorithms such as MU and HALS.

Recently with the advent of large scale internet data and interest in Big Data, researchers have started studying scalability of many foundational machine learning algorithms. 
To illustrate the dimension of matrices commonly used in the machine learning community, we present a few examples. 
Now-a-days the adjacency matrix of a billion-node social network is common. 
In the matrix representation of a video data, every frame contains three matrices for each RGB color, which is reshaped into a column.  
Thus in the case of a 4K video, every frame will take approximately 27 million rows (4096 row pixels x 2196 column pixels x 3 colors). 
Similarly, the popular representation of documents in text mining is a bag-of-words matrix, where the rows are the dictionary and the columns are the documents (e.g., webpages). 
Each entry $A_{ij}$ in the bag-of-words matrix is generally the frequency count of the word $i$  in the document $j$. 
Typically with the explosion of the new terms in social media, the number of words spans to millions. 

To handle such high dimensional matrices, it is important to study low rank approximation methods in a data-distributed environment.
For example, in many large scale scenarios, data samples are collected and stored over many general purpose computers, as the set of samples is too large to store on a single machine. 
In this case, the computation must also be distributed across processors.
Local computation is preferred as local access of data is much faster than remote access due to the costs of interprocessor communication. 
However, for low rank approximation algorithms that depend on global information (like MU, HALS, and BPP), some communication is necessary.

The simplest way to organize these distributed computations on large datasets is through a MapReduce framework like Hadoop, but this simplicity comes at the expense of performance.
In particular, most MapReduce frameworks require data to be read from and written to disk at every iteration, and they involve communication-intensive global, input-data shuffles across machines.

In this work, we present a much more efficient algorithm and implementation using tools from the field of High-Performance Computing (HPC).
We maintain data in memory (distributed across processors), take advantage of optimized libraries for local computational routines, and use the Message Passing Interface (MPI) standard to organize interprocessor communication.
The current trend for high-performance computers is that available parallelism (and therefore aggregate computational rate) is increasing much more quickly than improvements in network bandwidth and latency.
This trend implies that the relative cost of communication (compared to computation) is increasing.

To address this challenge, we analyze algorithms in terms of both their computation and communication costs.
In particular, our proposed algorithm ensures that after the input data is initially read into memory, it is \emph{never} communicated; we communicate only the factor matrices and other smaller temporary matrices.
Furthermore, we prove that in the case of dense input and under the assumption that $k\leq \sqrt{mn/p}$, our proposed algorithm \emph{minimizes} bandwidth cost (the amount of data communicated between processors) to within a constant factor of the lower bound.
We also reduce latency costs (the number of times processors communicate with each other) by utilizing MPI collective communication operations, along with temporary local memory space, performing $O(\log p)$ messages per iteration, the minimum achievable for aggregating global data. 

Fairbanks et al. \cite{Fairbanks2015} discuss a parallel NMF algorithm designed for multicore machines.  
To demonstrate the importance of minimizing communication, we consider this approach to parallelizing an alternating NMF algorithm in distributed memory.
While this naive algorithm exploits the natural parallelism available within the alternating iterations (the fact that rows of $\WW$ and columns of $\HH$ can be computed independently), it performs more communication than necessary to set up the independent problems.
We compare the performance of this algorithm with our proposed approach to demonstrate the importance of designing algorithms to minimize communication; that is, simply parallelizing the computation is not sufficient for satisfactory performance and parallel scalability.

The main contribution of this work is a new, high-performance parallel algorithm for non-negative matrix factorization.
The algorithm is flexible, as it is designed for both sparse and dense input matrices and can leverage many different algorithms for solving local non-negative least squares problems.
The algorithm is also efficient, maintaining data in memory, using MPI collectives for interprocessor communication, and using efficient libraries for local computation.
Furthermore, we provide a theoretical communication cost analysis to show that our algorithm reduces communication relative to the naive approach, and in the case of dense input, that it provably minimizes communication.
We show with performance experiments that the algorithm outperforms the naive approach by significant factors, and that it scales well for up to 100s of processors on both synthetic and real-world data.

 \section{Preliminaries}

\subsection{Notation}
\label{sec:notations}

Table \ref{tab:notation} summarizes the notation we use throughout.
We use \emph{upper case} letters for matrices and \emph{lower case} letters for vectors.
For example, $\AA$ is a matrix and $\aa$ is a column vector and $\aa^T$ is a row vector. 
The subscripts are used for sub-blocks of matrices. 
We use $m$ and $n$ to denote the numbers of rows and columns of $\AA$, respectively, and we assume without loss of generality $m>n$ throughout.

\begin{table}\begin{center}
\begin{tabular}{|l|l|}
\hline
$\AA$ & Input Matrix \\
$\WW$ & Left Low Rank Factor \\
$\HH$ & Right Low Rank Factor \\
$m$ & Number of rows of input matrix \\
$n$ & Number of columns of input matrix \\
$k$ & Low Rank \\
$\M{M}_i$ & $i$th row block of matrix $\M{M}$ \\
$\M{M}^i$ & $i$th column block of matrix $\M{M}$  \\
$\M{M}_{ij}$ & $(i,j)$th subblock of $\M{M}$ \\
$p$ & Number of parallel processes \\
$p_r$ & Number of rows in processor grid \\
$p_c$ & Number of columns in processor grid \\
\hline
\end{tabular}
\end{center}
\caption{Notation}
\label{tab:notation}
\end{table}
\subsection{Communication model}
\label{sec:comm-model}

To analyze our algorithms, we use the $\alpha$-$\beta$-$\gamma$ model of distributed-memory parallel computation.
In this model, interprocessor communication occurs in the form of messages sent between two processors across a bidirectional link (we assume a fully connected network).
We model the cost of a message of size $n$ words as $\alpha+n\beta$, where $\alpha$ is the per-message latency cost and $\beta$ is the per-word bandwidth cost.
Each processor can compute floating point operations (flops) on data that resides in its local memory; $\gamma$ is the per-flop computation cost.
With this communication model, we can predict the performance of an algorithm in terms of the number of flops it performs as well as the number of words and messages it communicates.
For simplicity, we will ignore the possibilities of overlapping computation with communication in our analysis.
For more details on the $\alpha$-$\beta$-$\gamma$ model, see \cite{TRG05,CH+07}.

\subsection{MPI collectives}
\label{sec:collectives}

Point-to-point messages can be organized into collective communication operations that involve more than two processors.
MPI provides an interface to the most commonly used collectives like broadcast, reduce, and gather, as the algorithms for these collectives can be optimized for particular network topologies and processor characteristics.
The algorithms we consider use the all-gather, reduce-scatter, and all-reduce collectives, so we review them here, along with their costs.
Our analysis assumes optimal collective algorithms are used (see \cite{TRG05,CH+07}), though our implementation relies on the underlying MPI implementation.

At the start of an all-gather collective, each of $p$ processors owns data of size $n/p$. 
After the all-gather, each processor owns a copy of the entire data of size $n$. 
The cost of an all-gather is $\alpha\cdot \log p + \beta \cdot \frac{p-1}{p}n$.
At the start of a reduce-scatter collective, each processor owns data of size $n$.
After the reduce-scatter, each processor owns a subset of the sum over all data, which is of size $n/p$.
(Note that the reduction can be computed with other associative operators besides addition.)
The cost of an reduce-scatter is $\alpha\cdot \log p + (\beta+\gamma) \cdot \frac{p-1}{p}n$.
At the start of an all-reduce collective, each processor owns data of size $n$.
After the all-reduce, each processor owns a copy of the sum over all data, which is also of size $n$.
The cost of an all-reduce is $2\alpha\cdot \log p + (2\beta+\gamma) \cdot \frac{p-1}{p}n$.
Note that the costs of each of the collectives are zero when $p=1$.

 \section{Related work}\label{sec:related}

In the data mining and machine learning literature there is an overlap between low rank approximations and matrix factorizations due to the nature of applications. 
Despite its name, Non-negative ``Matrix Factorization'' is really a low rank approximation.  
In this section, we discuss the various distributed NMF algorithms. 

The recent distributed NMF algorithms in the literature are \cite{liao2014cloudnmf},
\cite{liu2010distributed}, \cite{gemulla2011large}, \cite{Yin2014} and \cite{Faloutsos2014}. 
All of these works propose distributed NMF algorithms implemented using Hadoop. 
Liu, Yang,  Fan, He and Wang \cite{liu2010distributed} propose running Multiplicative Update (MU) for KL divergence, squared loss, and ``exponential'' loss functions. 
Matrix multiplication, element-wise multiplication, and element-wise division are the building blocks of the MU algorithm. 
The authors discuss performing these matrix operations effectively in Hadoop for sparse matrices. 
In similar directions, Liao, Zhang, Guan and Zhou \cite{liao2014cloudnmf} implement an open source Hadoop based MU algorithm and study its scalability on large-scale biological datasets. 
Also, Yin, Ghao and Zhang \cite{Yin2014} present a scalable NMF that can perform frequent updates, which aim to use the most recently updated data. 
Gemmula, Nijkamp, Erik, Haas and Sismanis \cite{gemulla2011large} propose a \emph{Generic algorithm} that works on different loss functions, often involving the distributed computation of the gradient. 
According to the authors, this formulation can also be extended to handle non-negative constraints. 
Similarly Faloutsos, Beutel, Xing and Papalexakis, \cite{Faloutsos2014} propose a distributed, scalable method for decomposing matrices, tensors, and coupled data sets through stochastic gradient descent on a variety of objective functions. 
The authors also provide an implementation that can enforce non-negative constraints on the factor matrices.

 \section{Foundations}

In this section, we will briefly introduce the Alternating Non-negative Least Squares (ANLS) framework, multiple methods for solving non-negative least squares problems (NLS) including Block Principal Pivoting (BPP), and a straightforward approach to parallelization of the framework. 

\subsection{Alternating Non-negative Least Squares}

According to the ANLS framework, we first partition the variables of the NMF problem into two blocks  $\WW$ and $\HH$.
Then we solve the following equations iteratively until a stopping criteria is satisfied. 
\SplitN{\label{eqn:two block}} {
\WW &\leftarrow \Argmin{\tilde \WW\geq 0}\NormBr{\AA-\tilde\WW\HH}_F^2,\\
\HH &\leftarrow  \Argmin{\tilde\HH\geq 0}\NormBr{\AA-\WW\tilde\HH}_F^2.
}

The optimizations sub-problem for $\WW$ and $\HH$ are NLS problems which can be solved by a number of methods from generic constrained convex optimization to active-set methods. 
Typical approaches use form the normal equations of the least squares problem (and then solve them enforcing the non-negativity constraint), which involves matrix multiplications of the factor matrices with the data matrix.
Algorithm \ref{alg:anlsnmf} shows this generic approach.

\begin{algorithm}
\caption{$[\WW,\HH] = \text{ANLS}(A,k)$}
\label{alg:anlsnmf}
\begin{algorithmic}[1]
\Require $\AA$ is an $m\times n$ matrix, $k$ is rank of approximation
\State Initialize $\HH$ with a non-negative matrix in $\Rn{n\times k}_+$.
\While{not converged}
  \State Update $\WW$ using $\HH \HH^T$ and $\AA \HH^T$
  \label{line:anls:W}
  \State Update $\HH$ using $\WW^T\WW$ and $\WW^T \AA$
  \label{line:anls:H}
\EndWhile
\end{algorithmic}
\end{algorithm}

The specifics of lines \ref{line:anls:W} and  \ref{line:anls:H} depend on the NLS method used.
For example, the update equations for MU \cite{seung2001algorithms} are

\SplitN{\label{eqn:muupdate}} {
w_{ij} &\leftarrow w_{ij} \frac{(\AA \HH^T)_{ij}}{(\WW \HH \HH^T)_{ij}}\\
h_{ij} &\leftarrow  h_{ij} \frac{(\WW^T \AA)_{ij}}{(\WW^T \WW \HH)_{ij}}.
} 
Note that after computing $\HH \HH^T$ and $\AA \HH^T$, the cost of updating $\WW$ is dominated by computing $\WW\HH\HH^T$, which is $2mk^2$ flops.
Given $\AA\HH^T$ and $\WW\HH\HH^T$, each entry of $\WW$ can be updated independently and cheaply.
Likewise, the extra computation cost of updating $\HH$ is $2nk^2$ flops.

HALS updates $\WW$ and $\HH$ by applying block coordinate descent on the columns of $\WW$ and rows of $\HH$ \cite{cichocki2009nonnegative}.
The update rules are
\SplitN{\label{eqn:halsupdate}} {
\ww^i &\leftarrow \frac{\lt[ (\AA\HH^T)^i - \sum_{\substack{l=1 \\ l \neq i}}^k (\HH\HH^T)_{li} \ww^l \rt]_{+}}{(\HH\HH^T)_{ii}} \\
\hh_i &\leftarrow \frac{\lt[ (\WW^T\AA)_i - \sum_{\substack{l=1 \\ l \neq i}}^k (\WW^T\WW)_{li} \hh_l \rt]_{+}}{(\WW^T\WW)_{ii}},
} 
where $\ww^i$ is the $i$th column of $\WW$ and $\hh_i$ is the $i$ row of $\HH$.
Note that the columns of $\WW$ and rows of $\HH$ are updated in order, so that the most up-to-date values are used in the update.
The extra computation is again $2(m+n)k^2$ flops for updating both $\WW$ and $\HH$. 

\subsection{Block Principal Pivoting}

In this paper, we focus on and use the block principal pivoting \cite{kim2011fast} method to solve the non-negative least squares problem, as it has the fastest convergence rate (in terms of number of iterations). 
We note that any of the NLS methods can be used within our parallel frameworks (Algorithms \ref{alg:naive} and \ref{alg:2D}).  

BPP is the state-of-the-art method for solving the NLS subproblems in Eq. \eqref{eqn:two block}.
The main sub-routine of BPP is the single right-hand side NLS problem
\SplitN{\label{eqn:single NLS}}{
\min_{\xx\geq 0} \|\CC\xx-\bb\|_2^2.
}

The Karush-Kuhn-Tucker (KKT) optimality condition for  Eq.~\eqref{eqn:single NLS} is as follows
\SplitS{\label{eqn:KKT}}{
\yy &= \CC^T \CC \xx - \CC^T \bb\\
\yy &\geq 0\\
\xx &\geq 0\\
\xx^T \yy & = 0.
}
The KKT condition \eqref{eqn:KKT} states that at optimality, the support sets (i.e., the non-zero elements)
of $\xx$ and $\yy$ are complementary to each other. Therefore, Eq.~\eqref{eqn:KKT} is an instance of
the \emph{Linear Complementarity Problem} (LCP) which arises frequently in quadratic programming.
When $k\ll\min(m,n)$, active set and active-set like methods are very suitable because most
computations involve matrices of sizes $m\times k, n\times k$, and $k\times k$ which are
small and easy to handle.

If we knew which indices correspond to nonzero values in the optimal solution, then computing it is an unconstrained least squares problem on these indices.
In the optimal solution, call the set of indices $i$ such that $x_i=0$ the active set, and let the remaining indices be the passive set. The BPP algorithm works to find this active set and passive set. Since the above problem is convex, the correct partition of the optimal solution will satisfy the KKT condition (Eq.~\eqref{eqn:KKT}).  The BPP algorithm greedily swaps indices between the active and passive set until finding a partition that satisfies the KKT condition. In the partition of the optimal solution, the values of the indices that belong to the active set will take zero. The values of the indices that belong to the passive set are determined by solving the least squares problem using normal equations restricted to the passive set. Kim, He and Park \cite{kim2011fast}, discuss the BPP algorithm in further detail. 
Because the algorithm is iterative, we define $C_\text{BPP}(k,c)$ as the cost of solving $c$ instances of Eq.~\eqref{eqn:single NLS} using BPP, given the $k\times k$ matrix $\CC^T\CC$ and $c$ columns of the form $\bb$. 

\subsection{Naive Parallel NMF Algorithm} 
\label{sec:naive}

In this section we present a naive parallelization of any NMF algorithm \cite{Fairbanks2015} that follows the ANLS framework (Algorithm \ref{alg:anlsnmf}).
Each NLS problem with multiple right-hand sides can be parallelized on the observation that the problems for multiple right-hand sides are independent from each other. 
That is, we can solve several instances of Eq.~\eqref{eqn:single NLS} independently for different $\bb$ where $\CC$ is fixed, which implies that we can optimize row blocks of $\WW$ and column blocks of $\HH$ in parallel. 

\begin{algorithm}[!t]
\caption{$[\WW,\HH] = \text{Naive-Parallel-NMF}(\AA,k)$}
\label{alg:naive}
\begin{algorithmic}[1]
\Require $\AA$ is an $m\times n$ matrix distributed both row-wise and column-wise across $p$ processors, $k$ is rank of approximation
\Require Local matrices: $\AA_{i}$ is $m/p\times n$, $\AA^{i}$ is $m\times n/p$, $\WW_i$ is $m/p\times k$, $\HH^i$ is $k\times n/p$
\State $p_i$ initializes $\HH^i$
\While{not converged}
	\Statex \textbf{/* Compute $\WW$ given $\HH$ */} 
	\State collect $\HH$ on each processor using all-gather
		\label{line:naive:allgatherH}
	\State $p_i$ computes $\WW_i \leftarrow \Argmin{\tilde \WW\geq 0} \|A_i - \tilde \WW \HH\|$
		\label{line:naive:computeW}
	\Statex \textbf{/* Compute $\HH$ given $\WW$ */} 
	\State collect $\WW$ on each processor using all-gather
		\label{line:naive:allgatherW}
	\State $p_i$ computes $\HH^i \leftarrow \Argmin{\tilde \HH\geq 0} \|A^i - \WW \tilde \HH\|$
		\label{line:naive:computeH}
\EndWhile
\Ensure $\displaystyle \WW, \HH \approx \Argmin{\M{\tilde W} \geq 0, \M{\tilde H} \geq 0} \|\AA- \M{\tilde W} \M{\tilde H}\|$
\Ensure $\WW$ is an $m\times k$ matrix distributed row-wise across processors, $\HH$ is a $k\times n$ matrix distributed column-wise across processors
\end{algorithmic}
\end{algorithm}

Algorithm \ref{alg:naive} presents a straightforward approach to setting up the independent subproblems.
Let us divide $\WW$ into row blocks $\WW_1, \ldots, \WW_p$ and $\HH$ into column blocks $\HH^1, \ldots, \HH^p$. 
We then double-partition the data matrix $\AA$ accordingly into row blocks $\AA_{1}, \ldots, \AA_p$ and column blocks $\AA^1, \ldots, \AA^p$ so that processor $i$ owns both $\AA_i$ and $\AA^i$ (see Figure \ref{fig:naive}).
With these partitions of the data and the variables, one can implement any ANLS algorithm in parallel, with only one communication step for each solve.

The computation cost of Algorithm \ref{alg:naive} depends on the local NLS algorithm.
For comparison with our proposed algorithm, we assume each processor uses BPP to solve the local NLS problems.
Thus, the computation at line \ref{line:naive:computeW} consists of computing $\AA^i\HH^T$, $\HH\HH^T$, and solving NLS given the normal equations formulation of rank $k$ for $m/p$ columns.
Likewise, the computation at line \ref{line:naive:computeH} consists of computing $\WW^T\AA$, $\WW^T\WW$, and solving NLS for $n/p$ columns.
In the dense case, this amounts to $4mnk/p+(m+n)k^2+C_\text{BPP}((m+n)/p,k)$ flops.
In the sparse case, processor $i$ performs $2(\nnz(\AA_i)+\nnz(\AA^i))k$ flops to compute $\AA^i\HH^T$ and $\WW^T\AA_i$ instead of $4mnk/p$.

The communication cost of the all-gathers at lines \ref{line:naive:allgatherH} and \ref{line:naive:allgatherW}, based on the expression given in Section \ref{sec:collectives} is $\alpha\cdot 2\log p + \beta\cdot (m+n)k$.
The local memory requirement includes storing each processor's part of matrices $\AA$, $\WW$, and $\HH$.
In the case of dense $\AA$, this is $2mn/p+(m+n)k/p$ words, as $\AA$ is stored twice; in the sparse case, processor $i$ requires $\nnz(\AA_i)+\nnz(\AA^i)$ words for the input matrix and $(m+n)k/p$ words for the output factor matrices.
Local memory is also required for storing temporary matrices $\WW$ and $\HH$ of size $(m+n)k$ words.

We summarize the algorithmic costs of Algorithm \ref{alg:naive} in Table \ref{tab:costs}.
This naive algorithm \cite{Fairbanks2015} has three main drawbacks: (1) it requires storing two copies of the data matrix (one in row distribution and one in column distribution) and both full factor matrices locally, (2) it does not parallelize the computation of $\HH\HH^T$ and $\WW^T\WW$ (each processor computes it redundantly), and (3) as we will see in Section \ref{sec:parNMF}, it communicates more data than necessary.

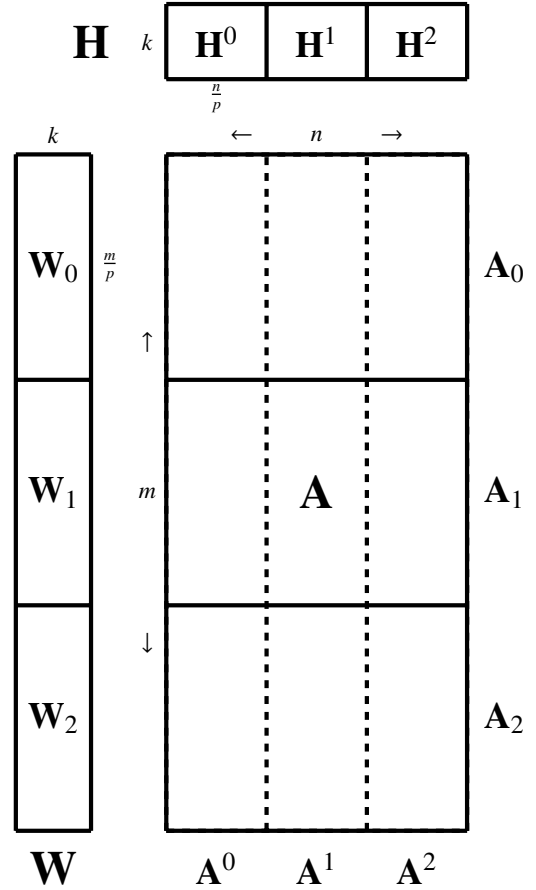
\begin{figure}[!t]
\centering
\begin{tikzpicture}

\draw[xscale=4/3,yscale=9,ultra thick, dashed] (0,0) grid (3,1);
\draw[xscale=4,yscale=3,ultra thick] (0,0) grid (1,3);
\draw[yscale=3,ultra thick] (-2,0) grid (-1,3);
\draw[xscale=4/3,ultra thick] (0,10) grid (3,11);

\node[draw=none] at (2,4.5) {\LARGE $\AA$};
\node[draw=none] at (4.5,7.5) {\Large $\AA_0$};
\node[draw=none] at (4.5,4.5) {\Large $\AA_1$};
\node[draw=none] at (4.5,1.5) {\Large $\AA_2$};
\node[draw=none] at (2/3,-0.5) {\Large $\AA^0$};
\node[draw=none] at (2,-0.5) {\Large $\AA^1$};
\node[draw=none] at (10/3,-0.5) {\Large $\AA^2$};
\node[draw=none] at (-1.5,-0.5) {\LARGE $\WW$};
\node[draw=none] at (-1.5,7.5) {\Large $\WW_0$};
\node[draw=none] at (-1.5,4.5) {\Large $\WW_1$};
\node[draw=none] at (-1.5,1.5) {\Large $\WW_2$};
\node[draw=none] at (-1,10.5) {\LARGE $\HH$};
\node[draw=none] at (2/3,10.5) {\Large $\HH^0$};
\node[draw=none] at (2,10.5) {\Large $\HH^1$};
\node[draw=none] at (10/3,10.5) {\Large $\HH^2$};

\node[draw=none] at (-0.25,10.5) {$k$};
\node[draw=none] at (-0.25,4.5) {$m$};
\node[draw=none] at (-0.25,6.5) {$\uparrow$};
\node[draw=none] at (-0.25,2.5) {$\downarrow$};
\node[draw=none] at (-0.75,7.5) {$\frac{m}{p}$};
\node[draw=none] at (-1.5,9.25) {$k$};
\node[draw=none] at (2,9.25) {$n$};
\node[draw=none] at (1,9.25) {$\leftarrow$};
\node[draw=none] at (3,9.25) {$\rightarrow$};
\node[draw=none] at (2/3,9.75) {$\frac{n}{p}$};

\end{tikzpicture} \caption{
Distribution of matrices for \NaiveAlg\ (Algorithm \ref{alg:naive}), for $p=3$.
Note that $\AA_i$ is $m/p \times n$, $\AA^i$ is $m\times n/p$, $\WW_i$ is $m/p_r \times k$, and $\HH^i$ is $k\times n/p$.
}
\label{fig:naive}
\end{figure}

 \section{High Performance Parallel NMF}
\label{sec:parNMF}

We present our proposed algorithm, \ParNMF, as Algorithm \ref{alg:2D}. 
The algorithm assumes a 2D distribution of the data matrix $\AA$ across a $p_r \times p_c$ grid of processors (with $p=p_rp_c$), as shown in Figure \ref{fig:2D-distribution}.
Algorithm \ref{alg:2D} performs an alternating method in parallel with a per-iteration bandwidth cost of $O\lt(\min\lt\{\sqrt{mnk^2/p},nk\rt\}\rt)$ words, latency cost of $O(\log p)$ messages, and load-balanced computation (up to the sparsity pattern of $\AA$ and convergence rates of local BPP computations).
To minimize the communication cost and local memory requirements, $p_r$ and $p_c$ are chosen so that $m/p_r\approx n/p_c\approx \sqrt{mn/p}$, in which case the bandwidth cost is $O\lt(\sqrt{mnk^2/p}\rt)$.

If the matrix is very tall and skinny, i.e., $m/p>n$, then we choose $p_r=p$ and $p_c=1$.
In this case, the distribution of the data matrix is 1D, and the bandwidth cost is $O(nk)$ words.

The matrix distributions for Algorithm \ref{alg:2D} are given in Figure \ref{fig:2D-distribution}; we use a 2D distribution of $\AA$ and 1D distributions of $\WW$ and $\HH$.
Recall from Table \ref{tab:notation} that $\M{M}_i$ and $\M{M}^i$  denote row and column blocks of $\M{M}$, respectively.
Thus, the notation $(\WW_i)_j$ denotes the $j$th row block within the $i$th row block of $\WW$.
Lines \ref{line:2DsyrkH}--\ref{line:2DcompW} compute $\WW$ for a fixed $\HH$, and lines \ref{line:2DsyrkW}--\ref{line:2DcompH} compute $\HH$ for a fixed $\WW$; note that the computations and communication patterns for the two alternating iterations are analogous.

In the rest of this section, we derive the per-iteration computation and communication costs, as well as the local memory requirements.
We also argue the communication-optimality of the algorithm in the dense case.
Table \ref{tab:costs} summarizes the results of this section and compares them to \NaiveAlg.

\begin{table*}\begin{center}
\begin{tabular}{|c|c|c|c|c|}
\hline
\textbf{Algorithm} & \textbf{Flops} & \textbf{Words} & \textbf{Messages} & \textbf{Memory} \\ \hline
\NaiveAlg & $O\lt(\frac{mnk}{p}+(m+n)k^2+C_\text{BPP}\lt(\frac{m+n}{p},k\rt)\rt)$ & $O((m+n)k)$ & $O(\log p)$ & $O\lt(\frac{mn}{p}+(m+n)k\rt)$ \\ \hline
\ParNMF\ ($m/p \geq n$) & $O\lt(\frac{mnk}{p}+C_\text{BPP}\lt(\frac{m+n}{p},k\rt)\rt)$ & $O(nk)$ & $O(\log p)$ & $O\lt(\frac{mn}{p}+\frac{mk}{p}+nk\rt)$ \\ \hline
\ParNMF\ ($m/p < n$) & $O\lt(\frac{mnk}{p}+C_\text{BPP}\lt(\frac{m+n}{p},k\rt)\rt)$ & $O\lt( \sqrt{\frac{mnk^2}{p}}\rt)$ & $O(\log p)$ & $O\lt(\frac{mn}{p}+\sqrt{\frac{mnk^2}{p}}\rt)$ \\ \hline
Lower Bound & $-$ & $\Omega(\min\lt\{\sqrt{\frac{mnk^2}{p}},nk\rt\})$ & $\Omega(\log p)$ & $\frac{mn}{p}+\frac{(m+n)k}{p}$ \\ \hline
\end{tabular}
\end{center}
\caption{Algorithmic costs for \NaiveAlg\ and \ParNMF\ assuming data matrix $\AA$ is dense.  Note that the communication costs (words and messages) also apply for sparse $\AA$.}
\label{tab:costs}
\end{table*}
\begin{algorithm}[t!]
\caption{$[\WW,\HH] = \text{\ParNMF}(\AA,k)$}
\label{alg:2D}
\begin{algorithmic}[1]
\Require $\AA$ is an $m\times n$ matrix distributed across a $p_r\times p_c$ grid of processors, $k$ is rank of approximation
\Require Local matrices: $\AA_{ij}$ is $m/p_r\times n/p_c$, $\WW_i$ is $m/p_r\times k$, $(\WW_i)_j$ is $m/p\times k$, $\HH_j$ is $k\times n/p_c$, and $(\HH_j)_i$ is $k\times n/p$
\State $p_{ij}$ initializes $(\HH_j)_i$
\While{not converged}
	\Statex \textbf{/* Compute $\WW$ given $\HH$ */} 
	\State $p_{ij}$ computes $\M{U}_{ij}=(\HH_j)_i{(\HH_j)_i}^T$
		\label{line:2DsyrkH}
	\State compute $\HH\HH^T {=} \sum_{i,j} \M{U}_{ij}$ using all-reduce across all procs
		\label{line:2Dall-reduceH}
		\Statex\Comment{$\HH\HH^T$ is $k\times k$ and symmetric}
	\State $p_{ij}$ collects $\HH_j$ using all-gather across proc columns
		\label{line:2Dall-gatherH}
	\State $p_{ij}$ computes $\M{V}_{ij}=\AA_{ij}\HH_j^T$
		\label{line:2DNEW}
		\Statex\Comment{$\M{V}_{ij}$ is $m/p_r \times k$}
	\State compute $(\AA\HH^T)_i {=} \sum_j \M{V}_{ij}$ using reduce-scatter across proc row to achieve row-wise distribution of $(\AA\HH^T)_i$
		\label{line:2Dreduce-scatterAHT}
		\Statex \Comment{$p_{ij}$ owns $m/p\times k$ submatrix $((\AA\HH^T)_i)_j$}
	\State $p_{ij}$ computes $(\WW_i)_j = \Argmin{\M{\tilde W} \geq 0} \lt\|\M{\tilde W}(\HH\HH^T) - ((\AA\HH^T)_i)_j \rt\|$
		\label{line:2DcompW}
	\Statex \textbf{/* Compute $\HH$ given $\WW$ */}
	\State $p_{ij}$ computes $\M{X}_{ij}={(\WW_i)_j}^T(\WW_i)_j$
		\label{line:2DsyrkW}
	\State compute $\WW^T\WW {=} \sum_{i,j} \M{X}_{ij}$ using all-reduce across all procs
		\label{line:2Dall-reduceW}
		\Statex\Comment{$\WW^T\WW$ is $k\times k$ and symmetric}
	\State $p_{ij}$ collects $\WW_i$ using all-gather across proc rows
		\label{line:2Dall-gatherW}
	\State $p_{ij}$ computes $\M{Y}_{ij}={\WW_i}^T\AA_{ij}$
		\label{line:2DNEH}
		\Statex\Comment{$\M{Y}_{ij}$ is $k\times n/p_c$}
	\State compute $(\WW^T\AA)^j = \sum_i \M{Y}_{ij}$ using reduce-scatter across proc columns to achieve column-wise distribution of $(\WW^T\AA)^j$
		\label{line:2Dreduce-scatterWTA}
		\Statex\Comment{$p_{ij}$ owns $k\times n/p$ submatrix $((\WW^T\AA)^j)^i$}
	\State $p_{ij}$ computes $(\HH^j)^i = \Argmin{\M{\tilde H} \geq 0} \lt\|(\WW^T\WW)\M{\tilde H} - ((\WW^T\AA)^j)^i \rt\|$
		\label{line:2DcompH}
\EndWhile
\Ensure $\displaystyle \WW, \HH \approx \Argmin{\M{\tilde W} \geq 0, \M{\tilde H} \geq 0} \|\AA- \M{\tilde W} \M{\tilde H}\|$
\Ensure $\WW$ is an $m\times k$ matrix distributed row-wise across processors, $\HH$ is a $k\times n$ matrix distributed column-wise across processors
\end{algorithmic}
\end{algorithm}

\begin{figure}
\centering
\begin{tikzpicture}

\draw[xscale=2,yscale=3,ultra thick] (0,0) grid (2,3);
\draw[yscale=3/2] (-2,0) grid (-1,6);
\draw[yscale=3,ultra thick] (-2,0) grid (-1,3);
\draw[xscale=2/3] (0,10) grid (6,11);
\draw[xscale=2,ultra thick] (0,10) grid (2,11);

\node[draw=none] at (2,-0.5) {\LARGE $\AA$};
\node[draw=none] at (1,7.5) {\Large $\AA_{00}$};
\node[draw=none] at (1,4.5) {\Large $\AA_{10}$};
\node[draw=none] at (1,1.5) {\Large $\AA_{20}$};
\node[draw=none] at (3,7.5) {\Large $\AA_{01}$};
\node[draw=none] at (3,4.5) {\Large $\AA_{11}$};
\node[draw=none] at (3,1.5) {\Large $\AA_{21}$};
\node[draw=none] at (-1.5,-0.5) {\LARGE $\WW$};
\node[draw=none] at (-2.5,7.5) {\Large $\WW_0$};
\node[draw=none] at (-2.5,4.5) {\Large $\WW_1$};
\node[draw=none] at (-2.5,1.5) {\Large $\WW_2$};
\node[draw=none] at (-1.5,8.25) {$(\WW_0)_0$};
\node[draw=none] at (-1.5,6.75) {$(\WW_0)_1$};
\node[draw=none] at (-1.5,5.25) {$(\WW_1)_0$};
\node[draw=none] at (-1.5,3.75) {$(\WW_1)_1$};
\node[draw=none] at (-1.5,2.25) {$(\WW_2)_0$};
\node[draw=none] at (-1.5,0.75) {$(\WW_2)_1$};
\node[draw=none] at (-1,10.5) {\LARGE $\HH$};
\node[draw=none] at (1,11.5) {\Large $\HH^0$};
\node[draw=none] at (3,11.5) {\Large $\HH^1$};
\node[draw=none] at (.33,10.5) {\scriptsize $(\HH^0)^0$};
\node[draw=none] at (1,10.5) {\scriptsize $(\HH^0)^1$};
\node[draw=none] at (1.66,10.5) {\scriptsize $(\HH^0)^2$};
\node[draw=none] at (2.33,10.5) {\scriptsize $(\HH^1)^0$};
\node[draw=none] at (3,10.5) {\scriptsize $(\HH^1)^1$};
\node[draw=none] at (3.66,10.5) {\scriptsize $(\HH^1)^2$};

\node[draw=none] at (-0.25,10.5) {$k$};
\node[draw=none] at (-0.25,4.5) {$m$};
\node[draw=none] at (-0.25,6.5) {$\uparrow$};
\node[draw=none] at (-0.25,2.5) {$\downarrow$};
\node[draw=none] at (-0.75,7.5) {$\frac{m}{p_r}$};
\node[draw=none] at (-0.75,8.25) {$\uparrow$};
\node[draw=none] at (-0.75,6.75) {$\downarrow$};
\node[draw=none] at (-0.75,2.25) {$\frac{m}{p}$};
\node[draw=none] at (-1.5,9.25) {$k$};
\node[draw=none] at (2,9.25) {$n$};
\node[draw=none] at (1,9.25) {$\leftarrow$};
\node[draw=none] at (3,9.25) {$\rightarrow$};
\node[draw=none] at (1,9.75) {$\frac{n}{p_c}$};
\node[draw=none] at (0.5,9.75) {$\leftarrow$};
\node[draw=none] at (1.5,9.75) {$\rightarrow$};
\node[draw=none] at (3,9.75) {$\frac{n}{p}$};

\end{tikzpicture} \caption{
Distribution of matrices for \ParNMF\ (Algorithm \ref{alg:2D}), for $p_r=3$ and $p_c=2$.
Note that $\AA_{ij}$ is $m/p_r \times m/p_c$, $\WW_i$ is $m/p_r \times k$, $(\WW_i)_j$ is $m/p\times k$, $\HH_j$ is $k\times n/p_c$, and $(\HH_j)_i$ is $k\times n/p$.
}
\label{fig:2D-distribution}
\end{figure}
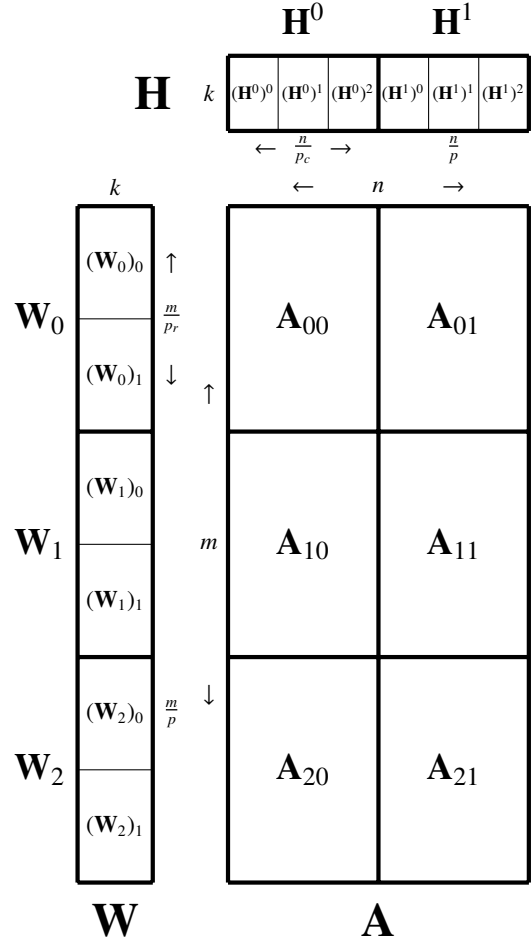

\paragraph{Computation Cost}

Local matrix computations occur at lines \ref{line:2DsyrkH}, \ref{line:2DNEW}, \ref{line:2DsyrkW}, and \ref{line:2DNEH}.
In the case that $\AA$ is dense, each processor performs 
$$\frac np k^2+2\frac{m}{p_r}\frac{n}{p_c}k+\frac mp k^2+2\frac{m}{p_r}\frac{n}{p_c}k=O\lt(\frac{mnk}{p}\rt)$$
flops.
In the case that $\AA$ is sparse, processor $(i,j)$ performs $(m+n)k^2/p$ flops in computing $\M{U}_{ij}$ and $\M{X}_{ij}$, and $4\nnz(\AA_{ij})k$ flops in computing $\M{V}_{ij}$ and $\M{Y}_{ij}$.
Local non-negative least squares problems occur at lines \ref{line:2DcompW} and \ref{line:2DcompH}.
In each case, the symmetric positive semi-definite matrix is $k\times k$ and the number of columns/rows of length $k$ to be computed are $m/p$ and $n/p$, respectively.
These costs together require $C_\text{BPP}(k,(m+n)/p)$ flops.
There are computation costs associated with the all-reduce and reduce-scatter collectives, both those contribute only to lower order terms.

\paragraph{Communication Cost}

Communication occurs during six collective operations (lines \ref{line:2Dall-reduceH}, \ref{line:2Dall-gatherH}, \ref{line:2Dreduce-scatterAHT}, \ref{line:2Dall-reduceW}, \ref{line:2Dall-gatherW}, and \ref{line:2Dreduce-scatterWTA}).
We use the cost expressions presented in Section \ref{sec:collectives} for these collectives.
The communication cost of the all-reduces (lines \ref{line:2Dall-reduceH} and \ref{line:2Dall-reduceW}) is $\alpha \cdot 4\log p+\beta \cdot 2k^2$; 
the cost of the two all-gathers (lines \ref{line:2Dall-gatherH} and \ref{line:2Dall-gatherW}) is $\alpha \cdot \log p + \beta \cdot \lt( (p_r{-}1)nk/p + (p_c{-}1)mk/p\rt)$; and
the cost of the two reduce-scatters (lines \ref{line:2Dreduce-scatterAHT} and  \ref{line:2Dreduce-scatterWTA}) is $\alpha \cdot \log p + \beta \cdot \lt( (p_c{-}1)mk/p + (p_r{-}1)nk/p\rt)$.

In the case that $m/p<n$, we choose $p_r=\sqrt{np/m} >1$ and $p_c=\sqrt{mp/n}>1$, and these communication costs simplify to $\alpha \cdot O(\log p) + \beta \cdot O(mk/p_r+nk/p_c+k^2) = \alpha \cdot O(\log p) + \beta \cdot O(\sqrt{mnk^2/p}+k^2)$.
In the case that $m/p\geq n$, we choose $p_c=1$, and the costs simplify to $\alpha \cdot O(\log p) + \beta \cdot O(nk)$.

\paragraph{Memory Requirements}

The local memory requirement includes storing each processor's part of matrices $\AA$, $\WW$, and $\HH$.
In the case of dense $\AA$, this is $mn/p+(m+n)k/p$ words; in the sparse case, processor $(i,j)$ requires $\nnz(\AA_{ij})$ words for the input matrix and $(m+n)k/p$ words for the output factor matrices.
Local memory is also required for storing temporary matrices $\WW_j$, $\HH_i$, $\M{V}_{ij}$, and $\M{Y}_{ij}$, of size $2mk/p_r+2nk/p_c)$ words.

In the dense case, assuming $k<n/p_c$ and $k<m/p_r$, the local memory requirement is no more than a constant times the size of the original data.
For the optimal choices of $p_r$ and $p_c$, this assumption simplifies to $k<\max\lt\{\sqrt{mn/p},m/p\rt\}$.

We note that if the temporary memory requirements become prohibitive, the computation of $((\AA \HH^T)_i)_j$ and $((\WW^T\AA)_j)_i$ via all-gathers and reduce-scatters can be blocked, decreasing the local memory requirements at the expense of greater latency costs.
While this case is plausible for sparse $\AA$, we did not encounter local memory issues in our experiments.

\paragraph{Communication Optimality}

In the case that $\AA$ is dense, Algorithm \ref{alg:2D} provably minimizes communication costs.
Theorem \ref{thm:LB} establishes the bandwidth cost lower bound for any algorithm that computes $\WW^T\AA$ or $\AA\HH^T$ each iteration.
A latency lower bound of $\Omega(\log p)$ exists in our communication model for any algorithm that aggregates global information \cite{CH+07}.
For NMF, this global aggregation is necessary each iteration to compute residual error in the case that $\AA$ is distributed across all $p$ processors, for example.
Based on the costs derived above, \ParNMF\ is communication optimal under the assumption $k<\sqrt{mn/p}$, matching these lower bounds to within constant factors.

\begin{theorem}[\cite{DE+13}]
\label{thm:LB}
Let $\AA \in \Rn{m \times n}$, $\WW \in \Rn{m \times k}$, and $\HH \in \Rn{k \times n}$ be dense matrices, with $k<n\leq m$.  If $k < \sqrt{mn/p}$, then any distributed-memory parallel algorithm on $p$ processors that load balances the matrix distributions and computes $\WW^T \AA$ and/or $\AA \HH^T$ must communicate at least $\Omega(\min\{\sqrt{mnk^2/p},nk\})$ words along its critical path.
\end{theorem}
\begin{proof}
The proof follows directly from \cite[Section II.B]{DE+13}.
Each matrix multiplication $\WW^T \AA$ and $\AA \HH^T$ has dimensions $k<n\leq m$, so the assumption $k<\sqrt{mn/p}$ ensures that neither multiplication has ``3 large dimensions.''
Thus, the communication lower bound is either $\Omega(\sqrt{mnk^2/p})$ in the case of $p>m/n$ (or ``2 large dimensions''), or $\Omega(nk)$, in the case of $p<m/n$ (or ``1 large dimension'').
If $p<m/n$, then $nk<\sqrt{mnk^2/p}$, so the lower bound can be written as $\Omega(\min\{\sqrt{mnk^2/p},nk\})$.
\end{proof}

We note that the communication costs of Algorithm \ref{alg:2D} are the same for dense and sparse data matrices (the data matrix itself is never communicated).
In the case that $\AA$ is sparse, this communication lower bound does not necessarily apply, as the required data movement depends on the sparsity pattern of $\AA$.
Thus, we cannot make claims of optimality in the sparse case (for general $\AA$).
The communication lower bounds for $\WW^T \AA$ and/or $\AA \HH^T$ (where $\AA$ is sparse) can be expressed in terms of hypergraphs that encode the sparsity structure of $\AA$ \cite{BDKS15}.
Indeed, hypergraph partitioners have been used to reduce communication and achieve load balance for a similar problem: computing a low-rank representation of a sparse tensor (without non-negativity constraints on the factors) \cite{KU15}.

 \section{Experiments}\label{sec:experiment}
In the data mining and machine learning community, there had been a large interest in using Hadoop for large scale implementation. Hadoop does lots of disk I/O and was designed for processing gigantic text files. Many of the real world data sets that is available for research are large scale sparse internet text data such as bag of words, recommender systems, social networks etc. Towards this end, there had been interest towards Hadoop implementation of matrix factorization algorithm \cite{gemulla2011large,liu2010distributed,liao2014cloudnmf}.  However, the use of NMF is beyond the sparse internet data and also applicable for dense real world data such as video, image etc.  Hence in order to keep our implementation applicable to wider audience, we chose MPI for distributed implementation.  Apart from the application point of view, we decided MPI C++ implementation for other technical advantages that is necessary for scientific application such as  (1) it can leverage the recent hardware improvements (2) effective communication between nodes (3) availability of numerically stable BLAS and LAPACK routines etc. We identified a few synthetic and real world datasets to experiment with our MPI implementation and a few baselines to compare our performance.  

\subsection{Experimental Setup}

\subsubsection{Datasets}

We used sparse and dense matrices that are synthetically generated and from real world. We will explain the datasets in this section.

\begin{itemize}
\item Dense Synthetic Matrix ({\em DSYN}): We generate a uniform random matrix of size 172,800 $\times$ 115,200 and add random Gaussian noise. The dimensions of this matrix is chosen such that it is uniformly distributable across processes. Every process will have its own prime seed that is different from other processes to generate the input random matrix $\AA$. 
\item Sparse Synthetic Matrix ({\em SSYN}): We generate a random sparse Erd\H{o}s-R\'{e}nyi matrix of the same dimensions 172,800 $\times$ 115,200 as the dense matrix, with density of 0.001.  That is, every entry is nonzero with probability 0.001.
\item Dense Real World Matrix ({\em Video}): Generally, NMF is performed in the video data for back ground subtraction to detect the moving objects. The low rank matrix $\hat{\AA} \approx \WW \HH^T$ represents background and the error matrix $\AA - \hat{\AA}$ has the moving objects.  Detecting moving objects has many real-world applications such as traffic estimation, security monitoring, etc. In the case of detecting moving objects, only the last minute or two of video is taken from the live video camera. The algorithm to incrementally adjust the NMF based on the new streaming video is presented in \cite{kim2013nonnegative}. To simulate this scenario, we collected a video in a busy intersection of the Georgia Tech campus at 20 frames per second for two minutes. We then reshaped the matrix such that every RGB frame is a column of our matrix, so that the matrix is dense with dimensions 1,013,400 $\times$ 2400.  \item Sparse Real World Matrix {\em Webbase} : We identified this dataset of a very large directed sparse graph with nearly 1 million nodes (1,000,005) and 3.1 million edges (3,105,536). The dataset was first reported by Williams et al. \cite{Williams2009}.  The NMF output of this directed graph will help us understand clusters in graphs. The size of both the real world datasets were adjusted to the nearest dimension for uniformly distributing the matrix.
\end{itemize}

\subsubsection{Machine}

We conducted our experiments on ``Edison'' at the National Energy Research Scientific Computing Center.
Edison is a Cray XC30 supercomputer with a total of 5,576 compute nodes, where each node has dual-socket 12-core Intel Ivy Bridge processors.
Each of the 24 cores has a clock rate of 2.4 GHz (translating to a peak floating point rate of 460.8 Gflops/node) and private 64KB L1 and 256KB L2 caches; each of the two sockets has a (shared) 30MB L3 cache; each node has 64 GB of memory.
Edison uses a Cray ``Aries'' interconnect that has a dragonfly topology.
Because our experiments use a relatively small number of nodes, we consider the local connectivity: a ``blade'' comprises 4 nodes and a router, and sets of 16 blades' routers are fully connected via a circuit board backplane (within a ``chassis'').
Our experiments do not exceed 64 nodes, so we can assume a very efficient, fully connected network.

\subsubsection{Initialization}

To ensure fair comparison among algorithms, the same random seed was used across different methods appropriately. 
That is, the initial random matrix $\HH$ was generated with the same random seed when testing with different algorithms (note that $\WW$ need not be initialized). 
This ensures that all the algorithms perform the same computations (up to roundoff errors), though the only computation with a running time that is sensitive to matrix values is the local NNLS solve via BPP.

\subsection{Algorithms}

For each of our data sets, we benchmark and compare three algorithms: (1) Algorithm \ref{alg:naive}, (2) Algorithm \ref{alg:2D} with $p_r=p$ and $p_c=1$ (1D processor grid), and (3) Algorithm \ref{alg:2D} with $p_r\approx p_c \approx \sqrt p$ (2D processor grid).
We choose these three algorithms to confirm the following conclusions from the analysis of Section \ref{sec:parNMF}: the performance of a naive parallelization of \NaiveAlg\ (Algorithm \ref{alg:naive}) will be severely limited by communication overheads, and the correct choice of processor grid within Algorithm \ref{alg:2D} is necessary to optimize performance.
To demonstrate the latter conclusion, we choose the two extreme choices of processor grids and test some data sets where a 1D processor grid is optimal (e.g., the Video matrix) and some where a squarish 2D grid is optimal (e.g., the Webbase matrix).

While we would like to compare against other high-performance NMF algorithms in the literature, the only other distributed-memory implementations of which we're aware are implemented using Hadoop and are designed only for sparse matrices \cite{liao2014cloudnmf},
\cite{liu2010distributed}, \cite{gemulla2011large}, \cite{Yin2014} and \cite{Faloutsos2014}.
We stress that Hadoop is not designed for high performance, requiring disk I/O between steps, so a run time comparison between a Hadoop implementation and a C++/MPI implementation is not a fair comparison of parallel algorithms.
To give a qualitative example of differences in run time, the running time of a Hadoop implementation of the MU algorithm on a large sparse matrix of dimension $2^{17} \times 2^{16}$ with $2 \times {10^8}$ nonzeros (with k=8) takes on the order of 50 minutes per iteration \cite{liu2010distributed}; our implementation takes a second per iteration for the synthetic data set (which is an order of magnitude larger in terms of rows, columns, and nonzeros) running on only 24 nodes. 

\subsection{Time Breakdown}

To differentiate the computation and communication costs among the algorithms, we present the time breakdown among the various tasks within the algorithms for both performance experiments.
For Algorithm \ref{alg:2D}, there are three local computation tasks and three communication tasks to compute each of the factor matrices:
\begin{itemize}
	\item \textbf{MM}, computing a matrix multiplication with the local data matrix and one of the factor matrices;
	\item \textbf{NLS}, solving the set of NLS problems using BPP;
	\item \textbf{Gram}, computing the local contribution to the Gram matrix;
	\item \textbf{All-Gather}, to compute the global matrix multiplication;
	\item \textbf{Reduce-Scatter}, to compute the global matrix multiplication;
	\item \textbf{All-Reduce}, to compute the global Gram matrix.
\end{itemize}
In our results, we do not distinguish the costs of these tasks for $\WW$ and $\HH$ separately; we report their sum, though we note that we do not always expect balance between the two contributions for each task.
Algorithm \ref{alg:naive} performs all of these tasks except the Reduce-Scatter and the All-Reduce; all of its communication is in the All-Gathers.

\begin{figure*}
     \centering
     	\subfloat[][Sparse Synthetic (SSYN) Comparison]{\includegraphics[width=3.5in,height=2in]{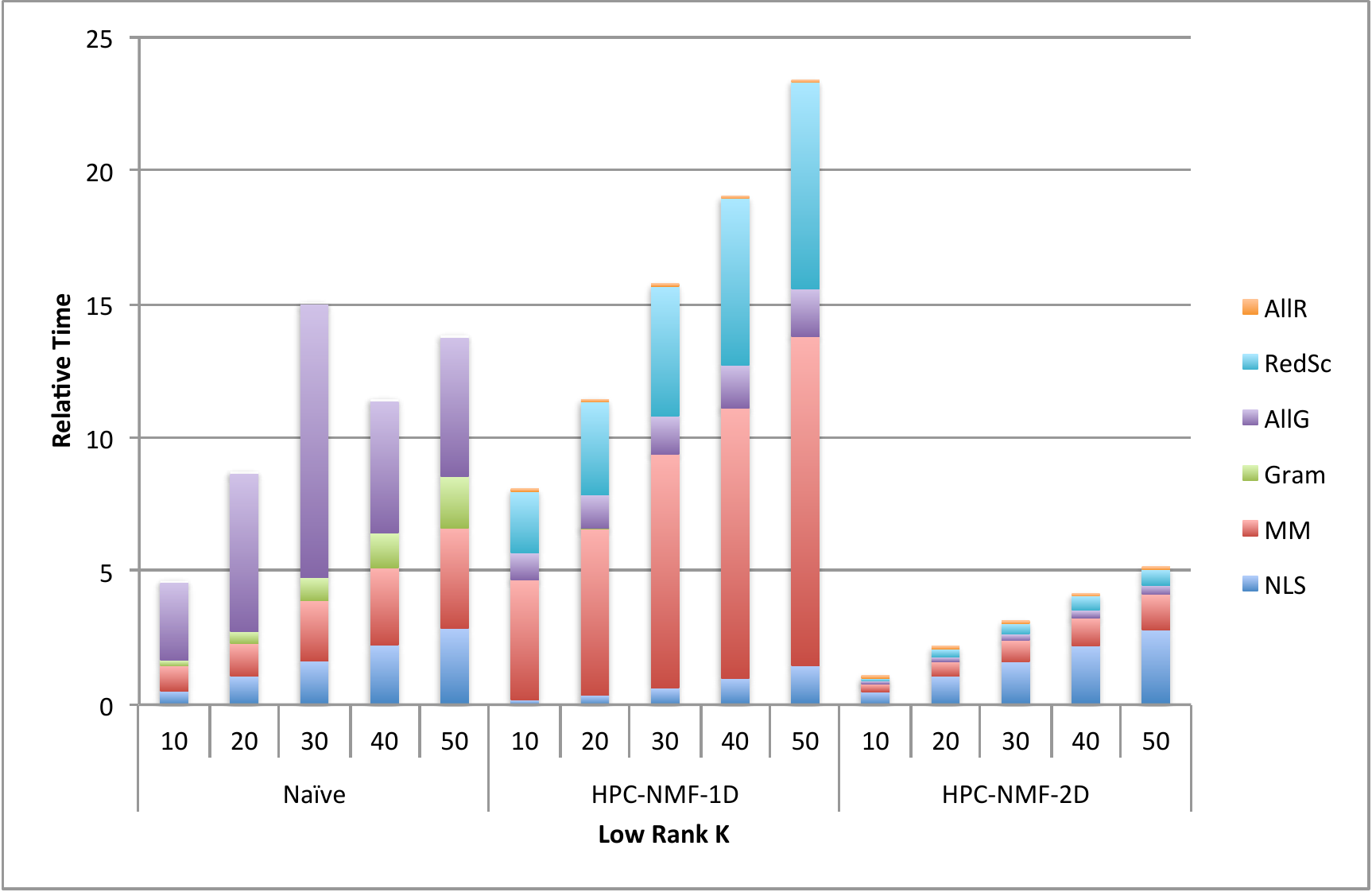}\label{fig:plots:ssyn-speedup}}
         \subfloat[][Sparse Synthetic (SSYN) Scaling]{\includegraphics[width=3.5in,height=2in]{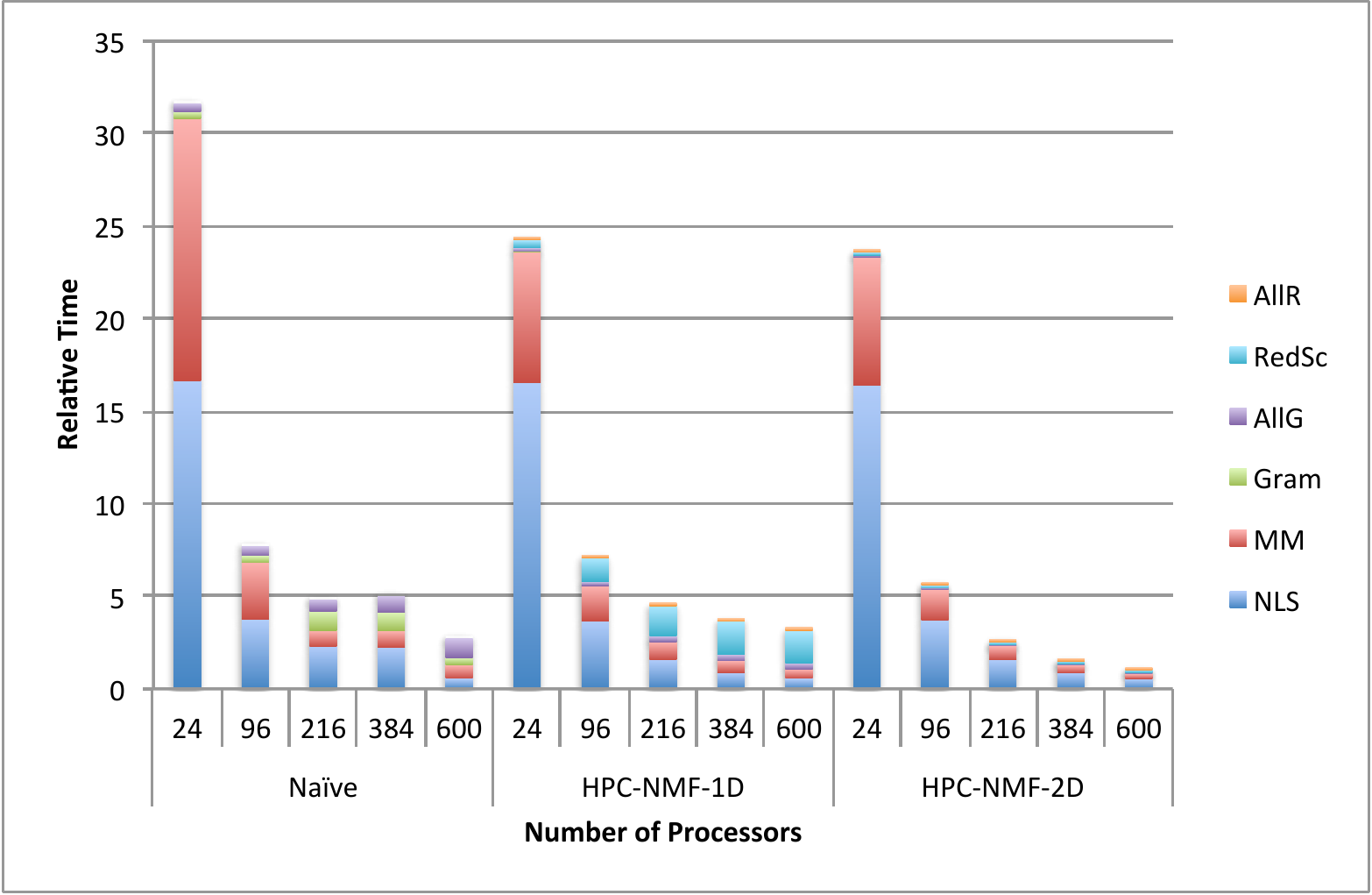}\label{fig:plots:ssyn-scaling}} \\
     \subfloat[][Dense Synthetic (DSYN) Comparison]{\includegraphics[width=3.5in,height=2in]{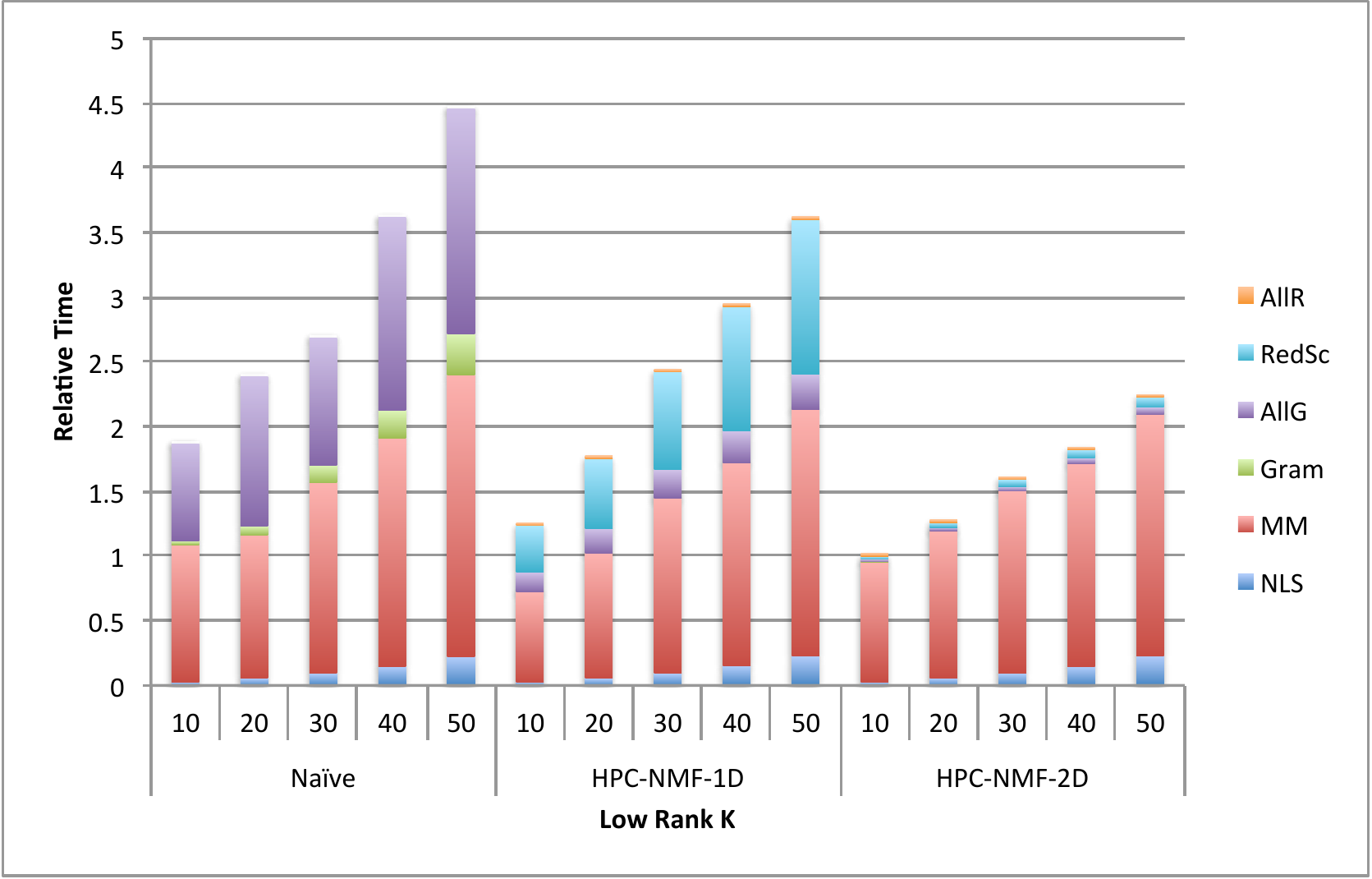}\label{fig:plots:dsyn-speedup}}
          \subfloat[][Dense Synthetic (DSYN) Scaling]{\includegraphics[width=3.5in,height=2in]{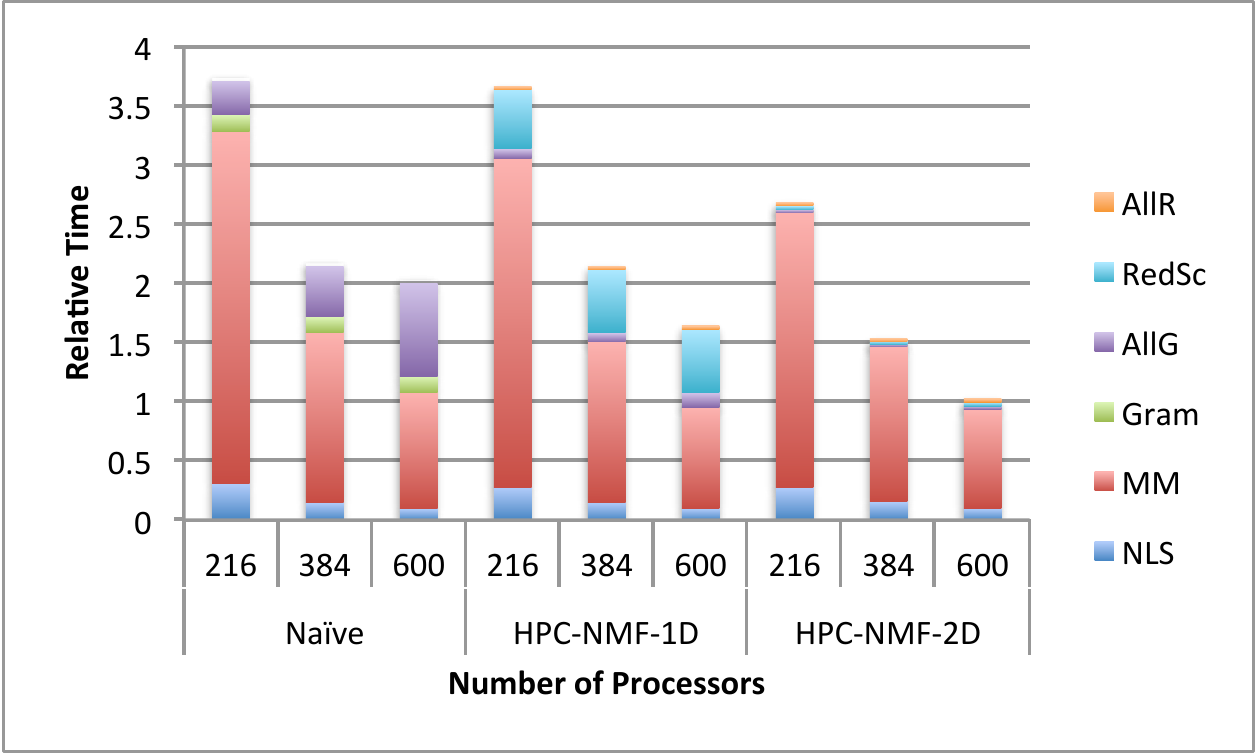}\label{fig:plots:dsyn-scaling}} \\
     \subfloat[][Webbase Comparison]{\includegraphics[width=3.5in,height=2in]{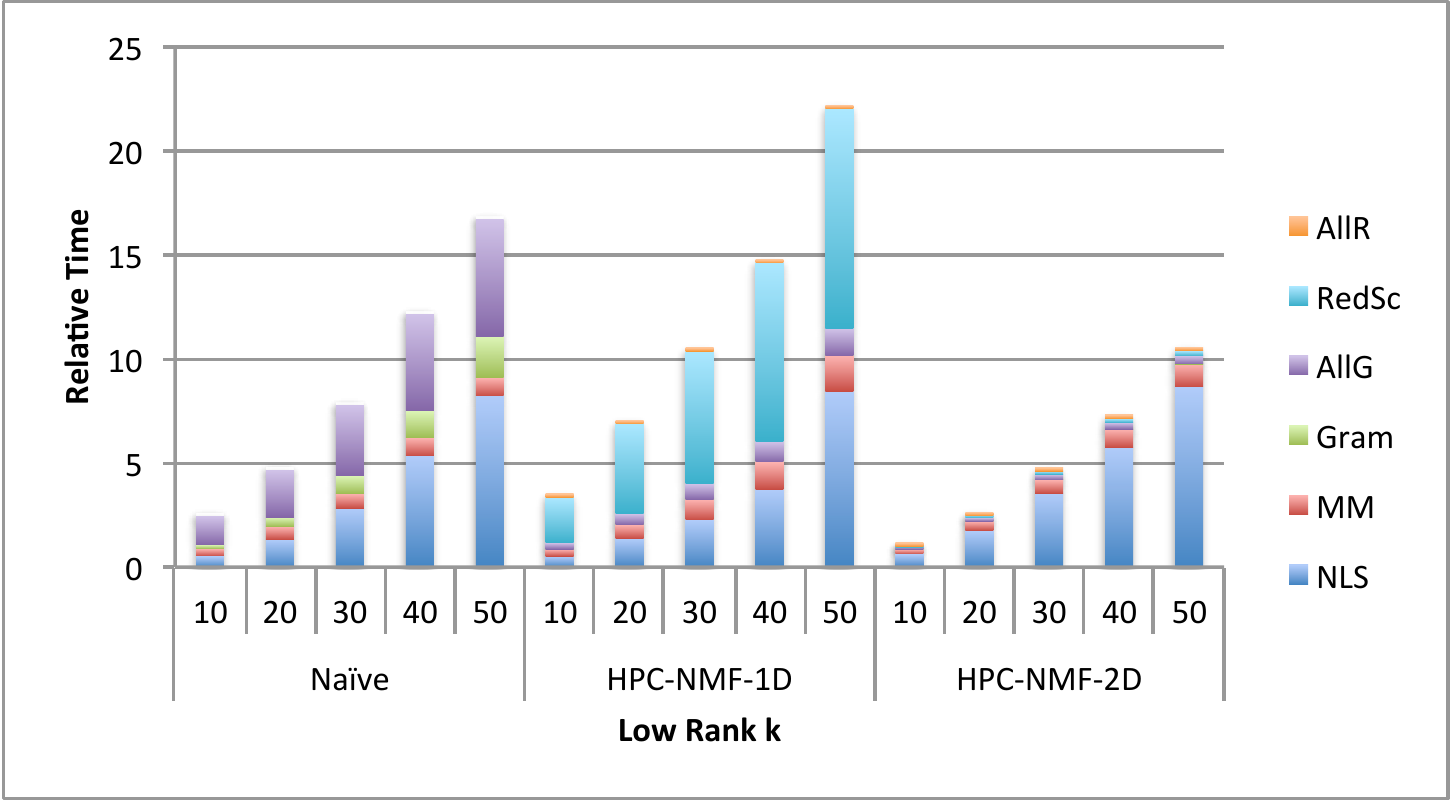}\label{fig:plots:webbase-speedup}}
          \subfloat[][Webbase Scaling]{\includegraphics[width=3.5in,height=2in]{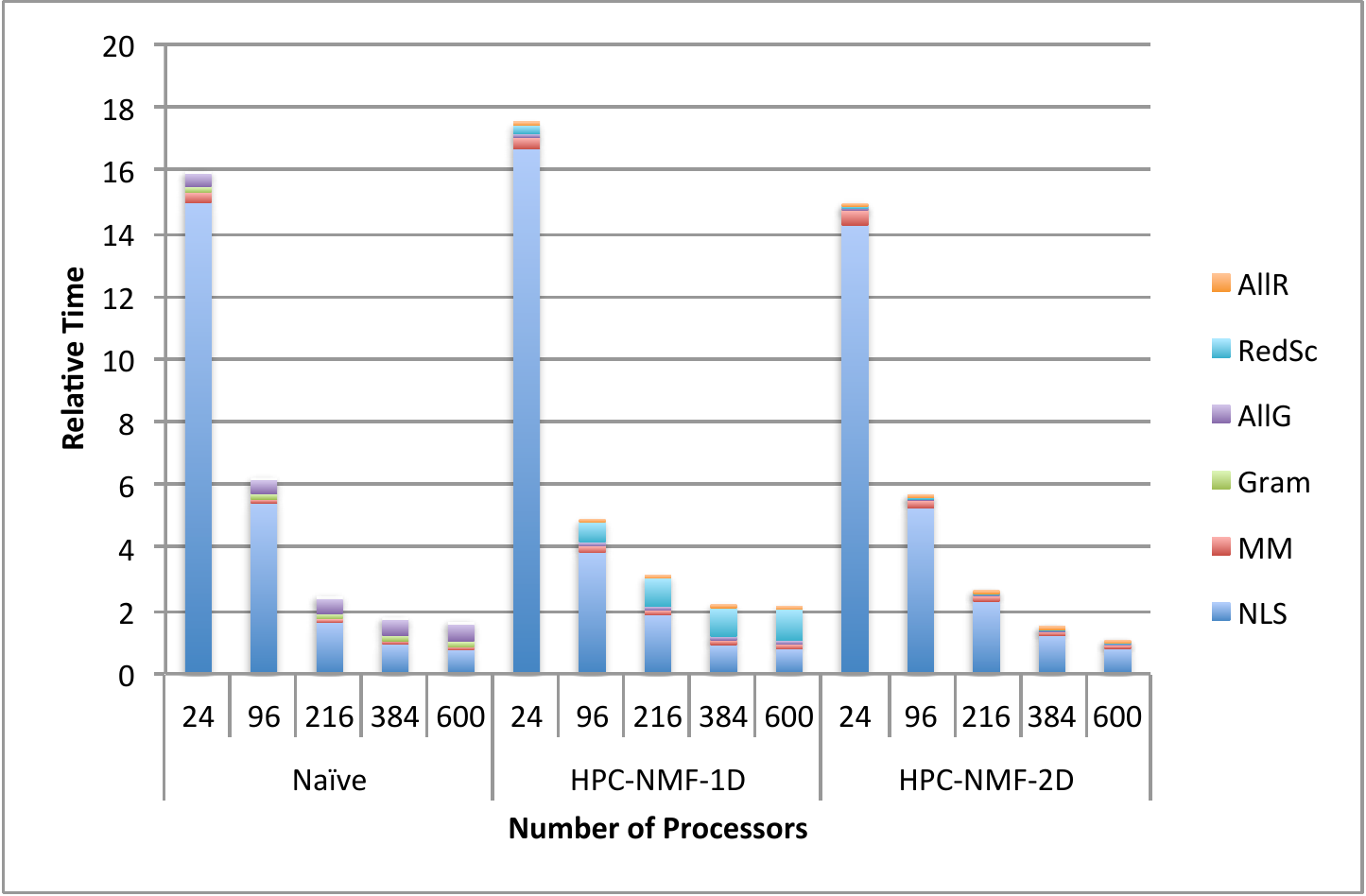}\label{fig:plots:webbase-scaling}} \\ 
     \subfloat[][Video Comparison]{\includegraphics[width=3.5in,height=2in]{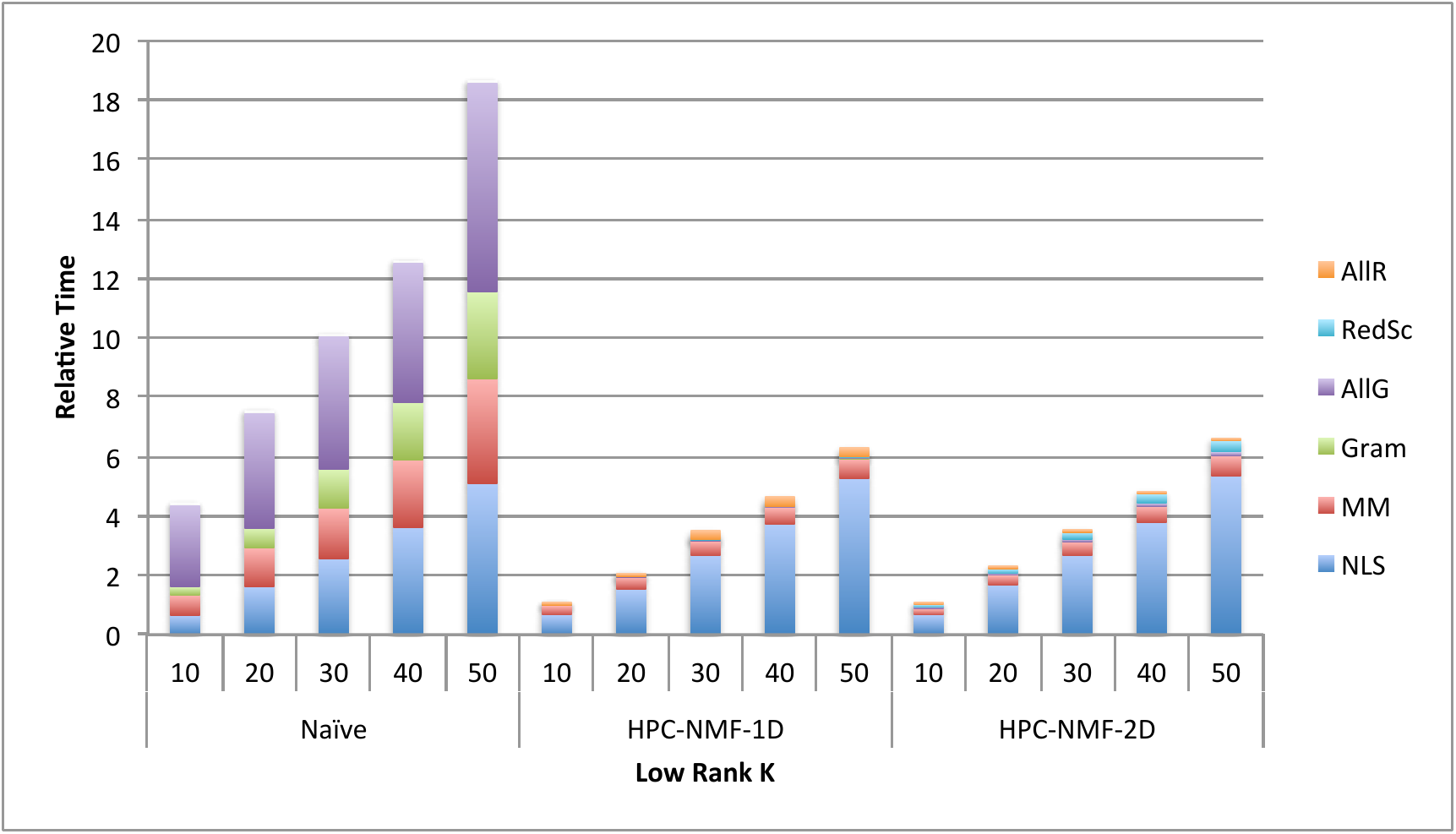}\label{fig:plots:video-speedup}} 
     \subfloat[][Video Scaling]{\includegraphics[width=3.5in,height=2in]{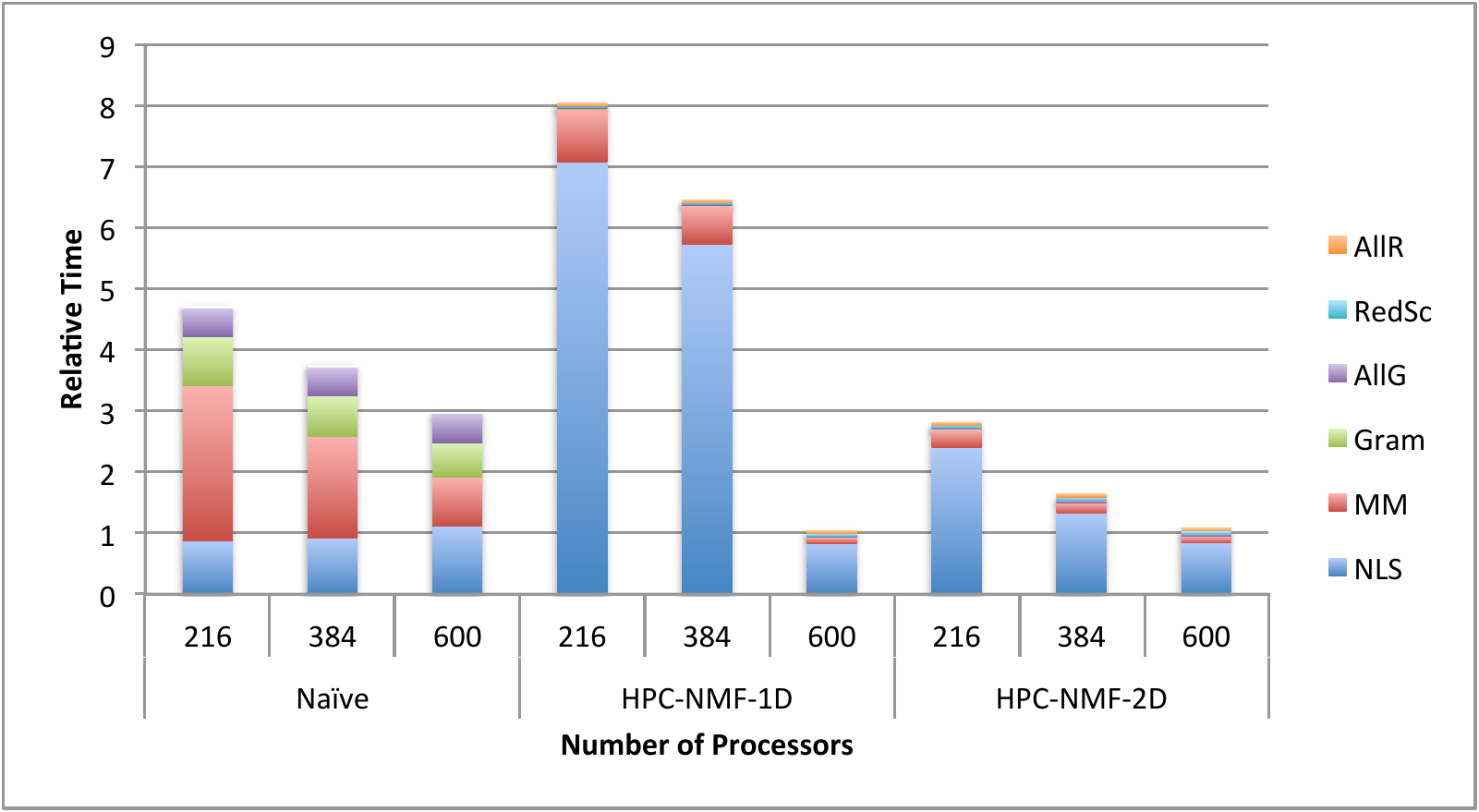}\label{fig:plots:video-scaling}} \\ 
     \caption{Experiments on Sparse and Dense Datasets}
     \label{fig:plots}
\end{figure*}

\subsection{Algorithmic Comparison}

Our first set of experiments is designed primarily to compare the three parallel implementations.
For each data set, we fix the number of processors to be 600 and vary the rank $k$ of the desired factorization.
Because most of the computation (except for NLS) and bandwidth costs are linear in $k$ (except for the All-Reduce), we expect linear performance curves for each algorithm individually.

The left side of Figure \ref{fig:plots} shows the results of this experiment for all four data sets.
The first conclusion we draw is that \ParNMF\ with a 2D processor grid performs significantly better than the \NaiveAlg; the largest speedup is $4.4\times$, for the sparse synthetic data and $k=10$ (a particularly communication bound problem).
Also, as predicted, the 2D processor grid outperforms the 1D processor grid on the squarish matrices.
While we expect the 1D processor grid to outperform the 2D grid for the tall-and-skinny Video matrix, their performance is comparable; this is because both algorithms are computation bound, as we see from Figure \ref{fig:plots:video-speedup}, so the difference in communication is negligible.

The second conclusion we can draw is that the scaling with $k$ tends to be close to linear, with an exception in the case of the Webbase matrix.
We see from Figure \ref{fig:plots:webbase-speedup} that this problem spends much of its time in NLS, which does not scale linearly with $k$.

We can also compare \ParNMF\ with a 1D processor grid with \NaiveAlg\ for squarish matrices (SSYN, DSYN, and Webbase).
Our analysis does not predict a significant difference in communication costs of these two approaches (when $m\approx n$), and we see from the data that \NaiveAlg\ outperforms \ParNMF\ for two of the three matrices (but the opposite is true for DSYN).
The main differences appear in the All-Gather versus Reduce-Scatter communication costs and the local MM (all of which are involved in the $\WW^T\AA$ computation).
In all three cases, our proposed 2D processor grid (with optimal choice of $m/p_r \approx n/p_c$) performs better than both alternatives.

\begin{table*}
\begin{tabular}{|c|cccc|cccc|cccc|}
\hline
&\multicolumn{4}{|c}{\NaiveAlg} & \multicolumn{4}{|c|}{\ParNMF-1D} & \multicolumn{4}{c|}{\ParNMF-2D} \\
\textbf{Cores} &\textbf{DSYN} & \textbf{SSYN} & \textbf{Video} & \textbf{Webbase} & \textbf{DSYN} & \textbf{SSYN} & \textbf{Video} & \textbf{Webbase} & \textbf{DSYN} & \textbf{SSYN} & \textbf{Video} & \textbf{Webbase} \\ \hline
24 &  &6.5632 & & 48.0256& &5.0821 & &52.8549 & &4.8427 & & 84.6286\\
96 &  &1.5929 & & 18.5507& & 1.4836& & 14.5873& & 1.1147& & 16.6966\\
216 &  2.1819 &0.6027 & 2.7899 & 7.1274 & 2.1548 & 0.9488 & 4.7928 & 9.2730 & 1.5283 & 0.4816 & 1.6106 & 7.4799 \\
384 &  1.2594 & 0.6466 & 2.2106 & 5.1431 & 1.2559 & 0.7695 & 3.8295 & 6.4740 & 0.8620 & 0.2661 & 0.8963 & 4.0630\\
600 &  1.1745 & 0.5592 & 1.7583 & 4.6825 & 0.9685 & 0.6666	& 0.5994 & 6.2751 & 0.5519 & 0.1683 & 0.5699 & 2.7376\\
\hline
\end{tabular}
\caption{Per-iteration running times of parallel NMF algorithms for $k=50$.}
\label{tab:exp}
\end{table*}

\subsection{Strong Scalability}

The goal of our second set of experiments is to demonstrate the (strong) scalability of each of the algorithms.
For each data set, we fix the rank $k$ to be 50 and vary the number of processors (this is a strong-scaling experiment because the size of the data set is fixed). 
We run our experiments on $\{24,96,216,384,600\}$ processors/cores, which translates to $\{1,4,9,16,25\}$ nodes.
The dense matrices are too large for 1 or 4 nodes, so we benchmark only on $\{216,384,600\}$ cores in those cases.

The right side of Figure \ref{fig:plots} shows the scaling results for all four data sets, and Table \ref{tab:exp} gives the overall per-iteration time for each algorithm, number of processors, and data set.
We first consider the \ParNMF\ algorithm with a 2D processor grid: comparing the performance results on 25 nodes (300 cores) to the 1 node (24 cores), we see nearly perfect parallel speedups.
The parallel speedups are $23\times$ for SSYN and $28\times$ for the Webbase matrix.
We believe the superlinear speedup of the Webbase matrix is a result of the running time being dominated by NLS; with more processors, the local NLS problem is smaller and more likely to fit in smaller levels of cache, providing better performance.
For the dense matrices, the speedup of \ParNMF\ on 25 nodes over 9 nodes is $2.7\times$ for DSYN and $2.8\times$ for Video, which are also nearly linear.

In the case of the \NaiveAlg\ algorithm, we do see parallel speedups, but they are not linear.
For the sparse data, we see parallel speedups of $10\times$ and $11\times$ with a $25\times$ increase in number of processors.
For the dense data, we observe speedups of $1.6\times$ and $1.8\times$ with a $2.8\times$ increase in the number of processors.
The main reason for not achieving perfect scaling is the unnecessary communication overheads.
 \section{Conclusion}\label{sec:conclusion}

In this paper, we propose a high-performance distributed-memory parallel algorithm that computes an NMF factorization by iteratively solving alternating non-negative least squares (NLS) subproblems.  
We show that by carefully designing a parallel algorithm, we can avoid communication overheads and scale well to modest numbers of cores.

For the datasets on which we experimented, we showed that an efficient implementation of a naive parallel algorithm spends much of its time in interprocessor communication.
In the case of \ParNMF, the problems remain computation bound on up to 600 processors, typically spending most of the time in local NLS solves.

We focus in this work on BPP, which is more expensive per-iteration than alternative methods like MU and HALS, because it has been shown to reduce overall running time in the sequential case by requiring fewer iterations \cite{kim2011fast}.
Because most of the time per iteration of \ParNMF\ is spent on local NLS, we believe further empirical exploration is necessary to confirm the advantages of BPP in the parallel case.
We note that if we use MU or HALS for local NLS, the relative cost of interprocessor communication will grow, making the communication efficiency of our algorithm more important.

In future work, we would like to extend this algorithm to dense and sparse tensors, computing the CANDECOMP/PARAFAC decomposition in parallel with non-negativity constraints on the factor matrices.
We would also like to explore more intelligent distributions of sparse matrices: while our 2D distribution is based on evenly dividing rows and columns, it does not necessarily load balance the nonzeros of the matrix, which can lead to load imbalance in MM.
We are interested in using graph and hypergraph partitioning techniques to load balance the memory and computation while at the same time reducing communication costs as much as possible.
Finally, we have not yet reached the limits of the scalability of \ParNMF; we would like to expand our benchmarks to larger numbers of nodes on the same size datasets to study performance behavior when communication costs completely dominate the running time.
 
\acks
This research was supported in part by an appointment to the Sandia National Laboratories Truman Fellowship in National Security Science and Engineering, sponsored by Sandia Corporation (a wholly owned subsidiary of Lockheed Martin Corporation) as Operator of Sandia National Laboratories under its U.S. Department of Energy Contract No. DE-AC04-94AL85000.

This research used resources of the National Energy Research Scientific Computing Center, a DOE Office of Science User Facility supported by the Office of Science of the U.S. Department of Energy under Contract No. DE-AC02-05CH11231.

\end{document}